\documentclass[letter,10pt]{article} 
\usepackage{color,amssymb,amsthm,amsmath,amsxtra,fullpage,graphicx,hyperref} 

\numberwithin{equation}{section}
\usepackage{pdflscape}

\usepackage{abstract}
\begin{document}

\def\draft{\centerline{\resizebox{!}{1.25cm}{\textcolor{red}{!! DRAFT !!}}}}
\def\bfS{{\textbf{S}}}
\def\mod{\text{ mod }}

\newcommand{\vcenteredinclude}[2]{\begingroup
\setbox0=\hbox{\includegraphics[#1]{#2}}%
\parbox{\wd0}{\box0}\endgroup}

\newcounter{constnum}
\setcounter{constnum}{0}

\def\constlabel#1{\newcounter{#1} \setcounter{#1}{\theconstnum}}

\def\const#1{%
\noindent\fbox{\parbox{\linewidth}{%
\refstepcounter{constnum}
{{\sc Construction \Roman{constnum}} \quad #1}}}}

\renewcommand{\qed}{\nobreak \ifvmode \relax \else
      \ifdim\lastskip<1.5em \hskip-\lastskip
      \hskip1.5em plus0em minus0.5em \fi \nobreak
      \vrule height0.75em width0.5em depth0.25em\fi}

\newtheorem{theorem}{Theorem}[section]
\newtheorem{lemma}[theorem]{Lemma}
\newtheorem{conjecture}[theorem]{Conjecture}
\newtheorem{proposition}[theorem]{Proposition}
\newtheorem{corollary}[theorem]{Corollary}

\renewcommand{\qed}{\nobreak \ifvmode \relax \else
      \ifdim\lastskip<1.5em \hskip-\lastskip
      \hskip1.5em plus0em minus0.5em \fi \nobreak
      \vrule height0.75em width0.5em depth0.25em\fi}

\theoremstyle{definition}
\newtheorem{example}[theorem]{Example}
\newtheorem{definition}[theorem]{Definition}	
\newtheorem{construction}[theorem]{Construction}

\centerline{{\LARGE Array Orthogonality in Higher Dimensions}}
\medskip
\centerline{\large Sam Blake, Andrew Tirkel}
\smallskip
\centerline{\large \it School of Mathematical Sciences, Monash University, Australia}
\bigskip

\begin{abstract}
\noindent{\sc Abstract.}  We generalize the array orthogonality property for perfect autocorrelation
sequences to $n$-dimensional arrays. The generalized array orthogonality property is used
to derive a number of $n$-dimensional perfect array constructions. 
\end{abstract}

% \todo{Do constructions III, IV and V generalize to n-dimensions? If not, can I explain why they don't?}\\

\section{Introduction}
Heimiller \cite{Heimiller1961} and Frank \cite{Frank1962} introduced a construction for perfect sequences of length
$n^2$ over $n$ roots of unity. Heimiller proved the construction produces perfect sequences of prime lengths 
by relating the autocorrelation of the sequence to the autocorrelation and cross-correlation of the columns of 
an array 
{\it associated} with the sequence. Similarly, sequences constructed by Milewski \cite{Milewski1983} used the same
method to prove they were perfect. Mow \cite{Mow1993} introduced the {\it array orthogonality property}, which 
generalized the proofs of Heimiller and Milewski to an arbitrary perfect sequence which is constructed by 
enumerating row-by-row the array associated with the sequence. Recently, the author \cite{Blake2014} gave a
sequence construction which possess the {\it array orthogonality property}.\\

\section{Preliminaries}

The \textit{periodic cross-correlation} of the sequences, $\textbf{a} =
\left[a_0,a_1, \cdots, a_{n-1}\right]$ and $\textbf{b} =
\left[b_0,b_1, \cdots, b_{n-1}\right]$ for shift $\tau$ is defined as 
$$\theta_{\textbf{a},\textbf{b}}(\tau) = \sum_{i=0}^{n-1}a_i b_{i+\tau}^*,$$
where $i+\tau$ is computed modulo $n$. Two sequences are
\textit{orthogonal} if $\theta_{\textbf{a},\textbf{b}}(\tau) = 0$ for
all $\tau$. \\

The \textit{periodic autocorrelation} of a sequence,
\textbf{s} for shift $\tau$ is given by $\theta_{\textbf{s}}(\tau) = \theta_{\textbf{s},\textbf{s}}(\tau)$. For 
$\tau \neq 0 \mod n$, $\theta_{\textbf{s}}(\tau)$ is called an \textit{off-peak} autocorrelation. A sequence 
is {\it perfect} if all off-peak periodic autocorrelation values are zero.\\

For applications, long perfect binary sequences are desired. However,
the longest known perfect binary sequence is the length 4 sequence
$\left[1,1,1,-1\right]$. It is conjectured that longer perfect
binary sequences do not exist\cite[conj. 3.9,
pp. 49]{Mow1993}. Consequently, sequences over roots of unity have 
been investigated for the last 60 years. \\

An $N$-dimensional array, \textbf{S}, over $n$ roots of unity is defined as 
$$\textbf{S} = \left[S_{i_0, i_1, \cdots, i_{N-1}} \right] = \omega^{f(i_0,i_1,\cdots, i_{N-1})},$$
where $f(i_0,i_1,\cdots, i_{N-1})$ is an integer function and $\omega$
is a primitive $n^{\text{th}}$ root of unity, that is $\omega =
e^{2\pi\sqrt{-1}/n}$. \\

A sequence is simply a one-dimensional array.  The periodic 
cross-correlation of two $N$-dimensional arrays, \textbf{A} and
\textbf{B}, both of size $l_0 \times l_1 \times\cdots\times l_{N-1}$, for shift $s_0, s_1, \cdots, s_{N-1}$ is defined as 
$$\theta_{\textbf{A}, \textbf{B}}\left(s_0, s_1, \cdots,
  s_{N-1}\right) = \sum_{i_0 = 0}^{l_0 - 1}\sum_{i_1 = 0}^{l_1 - 1}
\cdots \sum_{i_{N-1} = 0}^{l_{N-1} - 1} A_{i_0,i_1, \cdots, i_{N-1}}
B_{i_0 + s_0, i_1+s_1, \cdots, i_{N-1} + s_{N-1}}^*.$$ Similarily, the
periodic autocorrelation of a $N$-dimensional array for shift $s_0,
s_1, \cdots, s_{N-1}$ is given by $\theta_{\textbf{A}}\left(s_0,
s_1, \cdots, s_{N-1}\right) = \theta_{\textbf{A},\textbf{A}}\left(s_0,
s_1, \cdots, s_{N-1}\right)$. $\theta_{\textbf{A}}(s_0,s_1, \cdots, s_{N-1})$ is called an
\textit{off-peak} autocorrelation if not all $s_i = 0 \mod l_i$. An array is {\it perfect} if all 
off-peak autocorrelations are zero. 

\section{The Array Orthogonality Property}

We begin with the array orthogonality property (AOP). Consider a sequence 
$\textbf{s} = [s_0, s_1, \cdots, s_{ld^2-1}]$, then we call 
$$\textbf{S} = [S_{i,j}] = 
\left[
\begin{array}{cccccc}
 s_0 & s_1 & s_2 & \cdots  & \cdots  & s_{d-1} \\
 s_d & s_{d+1} & s_{d+2} & \cdots  & \cdots  & s_{2d-1} \\
 \vdots  & \vdots  & \vdots  &   &  & \vdots  \\
\vdots  & \vdots  &  \vdots &   &  & \vdots  \\
s_{(l-1)d} & s_{(l-1)d+1} & s_{(l-1)d+2} & \cdots  & \cdots  & s_{d l-1}
\end{array}
\right]$$
the array {\it associated} with \textbf{s} for the {\it divisor} $d$. We use the notation $\textbf{S}[n]$ to denote the 
$n$-th column of \textbf{S}. \\

\begin{definition}[AOP]\cite{Mow1993}
A sequence $\textbf{s} = [s_0, s_1, \cdots, s_{ld^2-1}]$ has the AOP for the divisor $d$ if the array 
\textbf{S} associated with \textbf{s} has the following two properties:
\begin{enumerate}
\item For all $\tau$ and $j_0 \neq j_1 \mod d$: \quad $\theta_{\textbf{S}[j_0], \textbf{S}[j_1]}(\tau) = 0$. (That is,
 any two distinct columns of \textbf{S} are orthogonal.)
\item For all $\tau \neq 0 \mod l d$: \quad $\displaystyle\sum_{j=0}^{d-1}\theta_{\textbf{S}[j]}(\tau) = 0$. 
(That is, the columns of \textbf{S} form a set of periodic complementary sequences.)
\end{enumerate}
\end{definition}

\begin{example}
We show that the Frank-Heimiller sequence of length 16 over 4 roots of unity has the AOP
for the divisor $d=4$. The sequence, in index notation (that is, the mapping: $2\pi\sqrt{-1}s_n/4 \rightarrow s_n$), is given by 
$$\textbf{s}=\left[0, 0, 0, 0, 0, 1, 2, 3, 0, 2, 0, 2, 0, 3, 2, 1\right],$$ 
and the array, \textbf{S}, associated with \textbf{s} for the divisor 4 is given by
$$\textbf{S} = \left[
\begin{array}{cccc}
 0 & 0 & 0 & 0 \\
 0 & 1 & 2 & 3 \\
 0 & 2 & 0 & 2 \\
 0 & 3 & 2 & 1
\end{array}
\right].$$ The cross-correlation of all 6 distinct pairs of columns is given by 
$$\theta _{[0,0,0,0], [0,1,2,3]}=\theta _{[0,0,0,0],[0,2,0,2]}=\theta _{[0,0,0,0],[0,3,2,1]}=
\theta _{[0,1,2,3],[0,2,0,2]}=\theta _{[0,1,2,3],[0,3,2,1]}=\theta _{[0,2,0,2],[0,3,2,1]}=[0,0,0,0].$$
So all distinct pairs of columns of \textbf{S} are orthogonal. Thus, \textbf{s} satisfies the first condition of the AOP. 
We now compute the autocorrelation of all the columns of \textbf{S}: 
\begin{align*}
\theta_{[0,0,0,0]} &= [4,4,4,4]\\
\theta_{[0,1,2,3]} &= [4,4\sqrt{-1},-4,-4\sqrt{-1}]\\
\theta_{[0,2,0,2]} &= [4,-4,4,-4]\\
\theta_{[0,3,2,1]} &= [4, -4\sqrt{-1},-4, 4\sqrt{-1}]
\end{align*}
For each off-peak shift, the sum of the autocorrelations of all the columns of \textbf{S} is zero. Thus, 
\textbf{s} satisfies the second condition of the AOP. 
\end{example}

\begin{theorem}
\cite{Mow1993} Any sequence with the AOP is perfect. 
\end{theorem}

\begin{proof}
The periodic autocorrelation of a sequence \textbf{s}, of length $l d^2$ for shift $\tau$ is given by 
$$\theta_{\textbf{s}}(\tau) = \sum_{i=0}^{l d^2 - 1} s_i \, s_{i+\tau}^*.$$
Change coordinates, let $i = q d + r$, ($r < d$), and $\tau = q' d + r'$, ($r'<d$). Then we have 
\begin{align*}
\theta_{\textbf{s}}(q'd+r') &= \sum_{r=0}^{d - 1} \sum_{q=0}^{l d - 1} s_{q d+r} \, s_{(q+q') d+r+r'}^*\\
&=  \sum_{r=0}^{d - 1} \sum_{q=0}^{l d - 1} s_{q d+r} \, s_{\left(q+q'+\left\lfloor \frac{r+r'}{d}\right\rfloor\right) d + (r+r' \mod d)}^*\\
&= \sum_{r=0}^{d - 1} \sum_{q=0}^{l d - 1}  S_{q,r} \, S_{q+q'+\left\lfloor \frac{r+r'}{d}\right\rfloor, (r+r' \mod d)}^*\\
&= \sum_{r=0}^{d - 1} \theta_{\textbf{S}[r],\textbf{S}[r+r' \mod d]} \left(q' + \left\lfloor \frac{r+r'}{d}\right\rfloor \right).
\end{align*}
For $r'\neq 0$, condition 1 of the AOP implies $\theta_{\textbf{s}}(\tau) = 0$. Otherwise, for $r'=0$, condition 2 of the AOP 
implies $\theta_{\textbf{s}}(\tau) = 0$. 
\end{proof}

The Frank and Heimiller sequences were the first sequences constructed which possessed the AOP. \\

\const{\hypertarget{FrankConst}{}\cite{Heimiller1961}\cite{Frank1962}
We construct a sequence \textbf{s} of length $n^2$ over $n$ roots of unity. Let $\textbf{S}' = [S'_{i,j}] = \omega^{i j}$ be 
an $n\times n$ array where $\omega=e^{2 \pi \sqrt{-1}/n}$. The sequence \textbf{s} is constructed by enumerating 
row-by-row the array $\textbf{S}'$.}\\

Heimiller showed \textbf{s} is perfect by showing $\textbf{S}'$ had the 
AOP. Heimiller's construction had the restriction that $n$ be a prime number. Frank generalized the 
Heimiller construction by removing this restriction. Other sequence constructions with the AOP include 
Milewski sequences \cite{Milewski1983} and constructions by the author \cite{Blake2014}. \\

\const{\hypertarget{MilewskiConst}{}\cite{Milewski1983} We construct a perfect sequence, \textbf{s}, of length $m^{2k+1}$ over 
$m^{k+1}$ roots of unity, where $k \geq 1$. Let $\textbf{u} = [u_i]$ be a Chu sequence \cite{Chu1972} of length $m$. Let 
$\textbf{S}' = [S'_{i,j}] = u_{i \mod m} \, \omega^{i j}$ be a $m^{k+1} \times m^k$ array where $\omega = e^{2 \pi \sqrt{-1}/m^{k+1}}$. 
The sequence \textbf{s} is constructed by enumerating row-by-row the array $\textbf{S}'$.}\\

\bigskip

The following construction borrows elements of the constructions of Frank and Milewski. The idea of 
using a piecewise function within a perfect sequence construction was introduced by Liu and Fan \cite{Liu2004}. \\

\const{\cite{Blake2014} We construct a perfect sequence of length $4mn^{k+1}$ over $2mn^k$ roots of unity. Let 
$\textbf{S}' = [S'_{i,j}] = \omega^{\lfloor i(i+j)/n\rfloor}$ be a $2mn^{k+1} \times 2$ array over $2mn^k$ roots of unity,
where $\omega=e^{2\pi\sqrt{-1}/(2mn^k)}$. The 
sequence \textbf{s} is constructed by enumerating row-by-row the array $\textbf{S}'$.}\\

Sequence constructions which do not have the AOP include Chu sequences \cite{Chu1972} and Liu--Fan 
sequences \cite{Liu2004}. We use these constructions within the higher dimensional constructions. Sequences 
with the AOP are yet to be used within perfect sequence constructions. \\

\section{The Generalized Array Orthogonality Property}

The idea that the AOP may be used to construct arrays in higher dimensions was used by Blake et al \cite{Blake2012}. 
We now turn our attention to AOP in higher dimensions. Consider two dimensions, let 
$\textbf{A} = [A_{i,j}]$ be an $n\times m$ array. Then the array, $\textbf{A}'$ is given by
{\small
$$
\left[
\begin{array}{lcc}
 \left[
\begin{array}{lccl}
 A_{0,0} & A_{0,1} & \cdots  & A_{0,d-1} \\
 A_{1,0} & A_{1,1} & \cdots  & A_{1,d-1} \\
 \vdots  & \vdots  & \ddots & \vdots  \\
 A_{d-1,0} & A_{d-1,1} & \cdots  & A_{d-1,d-1}
\end{array}
\right] & \cdots  & \left[
\begin{array}{lccl}
 A_{0,m-d-1} & A_{0,m-d} & \cdots  & A_{0,m-1} \\
 A_{1,m-d-1} & A_{1,m-d} & \cdots  & A_{1,m-1} \\
 \vdots  & \vdots  & \ddots & \vdots  \\
 A_{d-1,m-d-1} & A_{d-1,m-d} & \cdots  & A_{d-1,m-1}
\end{array}
\right] \\
 \left[
\begin{array}{lccl}
 A_{d,0} & A_{d,1} & \cdots  & A_{d,d-1} \\
 \vdots  & \vdots  &   & \vdots  \\
   &   & \ddots &   \\
 A_{2 d-1,0} & A_{2 d-1,1} & \cdots  & A_{2 d-1,d-1}
\end{array}
\right] & \ddots  & \vdots  \\
 \qquad\qquad\qquad\qquad\vdots  &  & \vdots  \\
 \left[
\begin{array}{lccl}
 A_{n-d-1,0} & A_{n-d-1,1} & \cdots  & A_{n-d-1,d-1} \\
 A_{n-d,0} & \vdots  &   & \vdots  \\
   &   & \ddots &   \\
 A_{n-1,0} & A_{n-1,1} & \cdots  & A_{n-1,d-1}
\end{array}
\right] & \cdots  & \left[
\begin{array}{lccl}
 A_{n-d-1,m-d-1} & A_{n-d-1,m-d} & \cdots  & A_{n-d-1,m-1} \\
 A_{n-d,m-d-1} & A_{n-d,m-d} & \cdots  & A_{n-d,m-1} \\
 \vdots  & \vdots  & \ddots & \vdots  \\
 A_{n-1,m-d-1} & A_{n-1,m-d} & \cdots  & A_{n-1,m-1}
\end{array}
\right]
\end{array}
\right].
$$}
We call this array the array {\it associated} with \textbf{A} for the {\it divisor} $d$. We 
use the notation $\textbf{A}'[k,l]$ to index $A_{i,j,k,l}$ for all $i,j$. We now state the $n$-dimensional
generalization of array association.\\

\begin{definition}[Array association]
Let $\textbf{A}'$ the $2n$-dimensional array $\textbf{A}' = [{A'}_{i_0,i_1, \cdots, i_{2n-1}}]$, then the 
$n$-dimensional array \textbf{A} associated with $\textbf{A}'$ for the divisor $d$ is given by 
$$\textbf{A} = \left[ {A'}_{d i_0 + i_n, di_1 + i_{n+1}, \cdots, d i_{n-1} + i_{2n-1}} \right].$$
\end{definition}

We can now state the {\it generalized array orthogonality property} (GAOP).\\

\begin{definition}[GAOP]
An $n$--dimensional array, \textbf{A}, has the GAOP for the divisor $d$ if the $2n$--dimensional array $\textbf{A}'$ associated 
with \textbf{A} has the following properties:
\begin{enumerate}
\item For all $s_0, s_1, \cdots, s_{n-1}$ and for all $i_n, i_{n+1}, \cdots, i_{2n-1}, j_n, j_{n+1}, \cdots, j_{2n-1} \mod d$ such that\newline 
$(i_n, i_{n+1}, \cdots, i_{2n-1}) \neq (j_n, j_{n+1}, \cdots, j_{2n-1})$: \quad
$$\theta_{\textbf{A}'[i_n, i_{n+1}, \cdots, i_{2n-1}], \textbf{A}'[j_n, j_{n+1}, \cdots, j_{2n-1}]}(s_0, s_1, \cdots, s_{n-1}) = 0.$$ (That is, all distinct 
$n$--dimensional arrays of $\textbf{A}'$ are orthogonal.)
\item For $s_0, s_1, \cdots, s_{n-1}\mod d$ such that not all $s_i = 0 \mod d$ (off-peak autocorrelation): \quad 
	$$\displaystyle\sum_{i_0=0}^{d-1}\sum_{i_1=0}^{d-1}\cdots\sum_{i_{n-1}=0}^{d-1}\theta_{\textbf{A}'[i_n, i_{n+1}, \cdots, i_{2n-1}]}
	(s_0, s_1, \cdots, s_{n-1}) = 0.$$ (That is, all the arrays $\textbf{A}'[i_n, i_{n+1}, \cdots, i_{2n-1}]$ form a set of periodic complementary arrays.)
\end{enumerate}
\end{definition}

We now state and prove our main theorem. \\

\begin{theorem}
Any $n$--dimensional array with the GAOP is perfect.
\end{theorem}	

\begin{proof}
Consider the autocorrelation of the array $\textbf{A} = [A_{i_0, i_1, \cdots, i_{n-1}}]$, with size $m_0\times m_1\times\cdots\times m_{n-1}$,  
$$\theta_{\textbf{A}}(s_0, s_1, \cdots, s_{n-1}) = \sum_{i_0=0}^{m_0-1}\sum_{i_1=0}^{m_1-1}\cdots
	\sum_{i_{n-1}=0}^{m_{n-1}-1} A_{q_0 d+r_0, \cdots, q_{n-1}d+r_{n-1}} \, A_{i_0+s_0, i_1+s_1, \cdots, i_{n-1}+s_{n-1}}^*$$
Introduce the change of variables $i_k = q_k d + r_k$, ($r_k<d$), then we have 
$$\theta_{\textbf{A}}(s_0, s_1, \cdots, s_{n-1}) = \sum_{r_0=0}^{d-1}\cdots\sum_{r_{d-1}=0}^{d-1}
	\sum_{q_0=0}^{m_0/d-1}\cdots\sum_{q_{n-1}=0}^{m_{n-1}/d-1}
	A_{q_0 d+r_0, \cdots, q_{n-1}d+r_{n-1}} \, A_{q_0 d+r_0+s_0, \cdots, q_{n-1}d+r_{n-1}+s_{n-1}}^*.$$
As before, introduce the change of variables $s_k = q'_k d + r'_k$, ($r'_k<d$), then we have
\begin{align*}
\theta_{\textbf{A}}(s_0, s_1, \cdots, s_{n-1}) &= \sum_{r_0=0}^{d-1}\cdots\sum_{r_{n-1}=0}^{d-1}\sum_{q_0=0}^{m_0/d-1}\cdots\sum_{q_{n-1}=0}^{m_{n-1}/d-1}
	A_{q_0 d+r_0, \cdots, q_{n-1}d+r_{n-1}} \, A_{(q_0 + q'_0) d+ r_0+r'_0, \cdots, (q_{n-1} + q'_{n-1}) d+ r_{n-1}+r'_{n-1}}^*\\
&= \sum_{r_0=0}^{d-1}\cdots\sum_{r_{n-1}=0}^{d-1}\sum_{q_0=0}^{m_0/d-1}\cdots\sum_{q_{n-1}=0}^{m_{n-1}/d-1}
A_{q_0 d+r_0, \cdots, q_{n-1}d+r_{n-1}} \times\\
&\qquad\qquad A_{\left(q_0 + q'_0 \left\lfloor \frac{r_0 + r'_0}{d} \right\rfloor\right) d+ (r_0+r'_0 \mod d), 
		\cdots,\left(q_{n-1} + q'_{n-1} \left\lfloor \frac{r_{n-1} + r'_{n-1}}{d} \right\rfloor\right) d+ (r_{n-1}+r'_{n-1} \mod d)}^*.
\end{align*}
Let $\textbf{A}'$ be a $2n$--dimensional array with size $m_0/d\times m_1/d\times\cdots \times m_{n-1}/d \times d \times d\times\cdots\times d$. 
($\textbf{A}'$ the array \textit{associated} with \textbf{A}.) Then we have 
\begin{align*}
\theta_{\textbf{A}}(s_0, s_1, \cdots, s_{n-1}) &=\sum_{r_0=0}^{d-1}\cdots\sum_{r_{n-1}=0}^{d-1}\sum_{q_0=0}^{m_0/d-1}\cdots\sum_{q_{n-1}=0}^{m_{n-1}/d-1}  A'_{q_0, \cdots, q_{n-1}, r_0, \cdots, r_{n-1}} \times\\
&\qquad\quad{A'}_{q_0 + q'_0 \left\lfloor \frac{r_0 + r'_0}{d} \right\rfloor, \cdots, 
	q_{n-1} + q'_{n-1} \left\lfloor \frac{r_{n-1} + r'_{n-1}}{d} \right\rfloor, r_0 + r'_0 \mod d, \cdots, r_{n-1} + r'_{n-1} \mod d}^*\\
&= \sum_{r_0=0}^{d-1}\cdots\sum_{r_{n-1}=0}^{d-1}
		\theta_{{\bf A}'[r_0,  \cdots, r_{n-1}], {\bf A}'[r_0+r'_0 \mod d, \cdots, r_{n-1}+r'_{n-1} \mod d]} 
		\left(Q_0, Q_1, \cdots, Q_{n-1}\right),\\
\end{align*}
where $Q_k = q_k + q'_k\left\lfloor \frac{r_k+r_k}{d} \right\rfloor$, and $\textbf{A}'[i_0, i_1, \cdots, i_{n-1}]$ is the $n$--dimensional array 
$[A'_{i_0, i_1, \cdots, i_{2n-1}}]$ where $i_n, i_{n+1}, \cdots, i_{2n-1}$ are fixed for each array. When 
$r_0+r'_0 = r_1+r'_1 = \cdots = r_{n-1}+r'_{n-1}  = 0 \mod d$ condition 1 of the GAOP implies 
$\theta_{\textbf{A}}(s_0, s_1, \cdots, s_{n-1}) = 0$. Otherwise, condition 2 of the GAOP implies 
$\theta_{\textbf{A}}(s_0, s_1, \cdots, s_{n-1}) = 0$.
\end{proof}

Note that the divisor, $d$, does not have to be the same in each dimension. Furthermore, as is the case in 
one-dimension, the array $\textbf{A}'$ is perfect. \\

\begin{corollary}
The array $\textbf{A}\,'$ is perfect. 
\end{corollary}

\begin{proof}
The proof follows from the fact that \textbf{A} has the GAOP.
\end{proof}

\bigskip

We now show the value of the GAOP by stating a construction for perfect $m$--dimensional arrays which are constructed by
concatenating (perfect) $2m$--dimensional arrays. \\

\const{%
\constlabel{gfrank}\hypertarget{gfrankconst}
Let $$\textbf{S}' = [S'_{i_0, i_1, \cdots, i_{2m-1}}] = 
\omega^{\displaystyle\left(\prod _{n=m}^{2 m-1} i_n + \sum _{n=0}^{m-1} i_n\, i_{n+m}\right)}$$ 
be a $2m$--dimensional array of size $d\times d\times\cdots\times d$, where $\omega = e^{2 \pi \, \sqrt{-1}/d}$. Let 
\textbf{S} be the $m$--dimensional array of size $d^2 \times d^2\times\cdots\times d^2$, formed by concatenating 
the array $\textbf{S}'$.}

\bigskip

Construction \Roman{gfrank} can be thought of a multi--dimensional generalization of Heimiller--Frank sequences. For $m=1$, 
Construction \Roman{gfrank} produces Heimiller--Frank sequences. \\

\begin{example}\label{GAOPexample}
We show a $9\times9$ array, \textbf{S}, from Construction \Roman{gfrank} has the GAOP for the divisor $d=3$. The array
\textbf{S}, (in index notation), is given by 
$$\textbf{S} = \left[
\begin{array}{ccccccccc}
 0 & 2 & 1 & 1 & 1 & 1 & 2 & 0 & 1 \\
 2 & 2 & 2 & 0 & 1 & 2 & 1 & 0 & 2 \\
 1 & 2 & 0 & 2 & 1 & 0 & 0 & 0 & 0 \\
 1 & 0 & 2 & 2 & 2 & 2 & 0 & 1 & 2 \\
 1 & 1 & 1 & 2 & 0 & 1 & 0 & 2 & 1 \\
 1 & 2 & 0 & 2 & 1 & 0 & 0 & 0 & 0 \\
 2 & 1 & 0 & 0 & 0 & 0 & 1 & 2 & 0 \\
 0 & 0 & 0 & 1 & 2 & 0 & 2 & 1 & 0 \\
 1 & 2 & 0 & 2 & 1 & 0 & 0 & 0 & 0
\end{array}
\right],$$
and the array $\textbf{S}'$ associated with \textbf{S} for the divisor 3 is given by
$$\textbf{S}' = \left[
\begin{array}{ccc}
 \left[
\begin{array}{ccc}
 0 & 2 & 1 \\
 2 & 2 & 2 \\
 1 & 2 & 0
\end{array}
\right] & \left[
\begin{array}{ccc}
 1 & 1 & 1 \\
 0 & 1 & 2 \\
 2 & 1 & 0
\end{array}
\right] & \left[
\begin{array}{ccc}
 2 & 0 & 1 \\
 1 & 0 & 2 \\
 0 & 0 & 0
\end{array}
\right] \\
 \left[
\begin{array}{ccc}
 1 & 0 & 2 \\
 1 & 1 & 1 \\
 1 & 2 & 0
\end{array}
\right] & \left[
\begin{array}{ccc}
 2 & 2 & 2 \\
 2 & 0 & 1 \\
 2 & 1 & 0
\end{array}
\right] & \left[
\begin{array}{ccc}
 0 & 1 & 2 \\
 0 & 2 & 1 \\
 0 & 0 & 0
\end{array}
\right] \\
 \left[
\begin{array}{ccc}
 2 & 1 & 0 \\
 0 & 0 & 0 \\
 1 & 2 & 0
\end{array}
\right] & \left[
\begin{array}{ccc}
 0 & 0 & 0 \\
 1 & 2 & 0 \\
 2 & 1 & 0
\end{array}
\right] & \left[
\begin{array}{ccc}
 1 & 2 & 0 \\
 2 & 1 & 0 \\
 0 & 0 & 0
\end{array}
\right]
\end{array}
\right].$$
We show the arrays $\textbf{S}'[1,1]$, and $\textbf{S}'[0,2]$ are orthogonal. The arrays are given by 
$$\textbf{S}'[1,1] = \left[
\begin{array}{ccc}
 2 & 1 & 0 \\
 1 & 0 & 2 \\
 0 & 2 & 1
\end{array}
\right]\qquad \text{ and }\qquad 
\textbf{S}'[0,2] = \left[
\begin{array}{ccc}
 1 & 1 & 1 \\
 2 & 2 & 2 \\
 0 & 0 & 0
\end{array}
\right],$$
and their cross-correlation, for all shifts, is given by 
$$\theta_{\textbf{S}'[1,1], \textbf{S}'[0,2]} = \left[
\begin{array}{ccc}
 0 & 0 & 0 \\
 0 & 0 & 0 \\
 0 & 0 & 0
\end{array}
\right].$$
The full calculation of the cross-correlation of all 36 distinct pairs of arrays is given in Appendix I. The sum of the correlations
of all arrays $\textbf{S}'[i,j]$, for $0\leq i<3$ and $0\leq j < 3$ is given by 
$$\sum_{i,j} \theta_{\textbf{S}'[i,j]} = \left[
\begin{array}{ccc}
 81 & 0 & 0 \\
 0 & 0 & 0 \\
 0 & 0 & 0
\end{array}
\right].$$
So \textbf{S} satisfies the second condition of the GAOP.
\end{example}

\begin{theorem}\label{newcons1}
The array \textbf{S} from Construction \Roman{gfrank} is perfect. 
\end{theorem}

\begin{proof}
We show the array, $\textbf{S} = [S_{i_0, i_1, \cdots, i_{2m-1}}]$, is perfect by showing 
$\textbf{S}$ has the GAOP. Firstly, we show that all distinct $m$--dimensional arrays of $\textbf{S}'$ are orthogonal.
\begin{align*}
&\theta_{\textbf{S}'[i_m, i_{m+1}, \cdots, i_{2m-1}], \textbf{S}'[i'_m, i'_{m+1}, \cdots, i'_{2m-1}]}(s_0, s_1, \cdots, s_{m-1}) = \\
	&\sum_{i_0=0}^{d-1}\sum_{i_1=0}^{d-1} \cdots \sum_{i_{m-1}=0}^{d-1} 
	{S'}_{i_0, i_1, \cdots, i_{m-1},i_m, i_{m+1}, \cdots, i_{2m-1}} \, {S'}_{i_0+s_0, i_1+s_1, \cdots, i_{m-1}+s_{m-1}, i'_m, i'_{m+1}, \cdots, i'_{2m-1}}^*\\
&= \sum_{i_0=0}^{d-1}\sum_{i_1=0}^{d-1} \cdots \sum_{i_{m-1}=0}^{d-1} 
	\omega^{\left(\prod _{n=m}^{2 m-1} i_n + \sum _{n=0}^{m-1} i_n\, i_{n+m}\right)}
	\omega^{-\left(\prod _{n=m}^{2 m-1} i'_n + \sum _{n=0}^{m-1} \left((i_{n} + s_n)\, i'_{n+m}\right)\right)}\\
&= \left(\omega^{\prod_{n=m}^{2m-1}i_n - \prod_{n=m}^{2m-1}i'_n - \sum_{n=0}^{m-1}s_n\,i'_{n+m}}\right)
		\sum_{i_0=0}^{d-1}\sum_{i_1=0}^{d-1} \cdots \sum_{i_{m-1}=0}^{d-1} 
		\omega^{\sum_{n=0}^{m-1}(i_{n+m}-i'_{n+m})i_n}\\
&=  \left(\omega^{\prod_{n=m}^{2m-1}i_n - \prod_{n=m}^{2m-1}i'_n - \sum_{n=0}^{m-1}s_n\,i'_{n+m}}\right)
		\sum_{i_0=0}^{d-1}\sum_{i_1=0}^{d-1} \cdots \sum_{i_{m-1}=0}^{d-1} 
		\left(\prod_{n=0}^{m-1}\omega^{(i_{n+m}-i'_{n+m})i_n}\right)\\
&= \left(\omega^{\prod_{n=m}^{2m-1}i_n - \prod_{n=m}^{2m-1}i'_n - \sum_{n=0}^{m-1}s_n\,i'_{n+m}}\right)
	\left(\sum_{i_0=0}^{d-1}\omega^{(i_m - i'_{m})i_0}\right)\left(\sum_{i_1=0}^{d-1}\omega^{(i_{1+m} - i'_{1+m})i_1}\right) \times\cdots\\
&\qquad\qquad\qquad\times \left(\sum_{i_{m-1}=0}^{d-1}\omega^{(i_{2m-1} - i'_{2m-1})i_{m-1}}\right)
\end{align*}
The sums above are Gaussian sums, which are zero as $i_n \neq i'_n$ for all $m\leq n < 2m$. So \textbf{S} 
satisfies the first condition of the GAOP. We now show \textbf{S} satisfies the second condition of the GAOP. 
\begin{align*}
&\sum_{i=0}^{d-1}\sum_{i_1=0}^{d-1}\cdots\sum_{i_{m-1}=0}^{d-1} \theta_{\textbf{S}'[i_m, i_{m+1}, \cdots, i_{2m-1}]}
	(s_0, s_1, \cdots, s_{m-1})=\\
&\sum_{i=0}^{d-1}\sum_{i_1=0}^{d-1}\cdots\sum_{i_{2m-1}=0}^{d-1} {S'}_{i_0, i_1, \cdots, i_{2m-1}}\, 
	{S'}_{i_0+s_0, i_1+s_1, \cdots, i_{m-1}+s_{m-1}, i_m, i_{m+1}, \cdots, i_{2m-1}}^* \\
&= \sum_{i=0}^{d-1}\sum_{i_1=0}^{d-1}\cdots\sum_{i_{2m-1}=0}^{d-1} 
	\omega^{\left(\prod_{n=m}^{2m-1}i_n + \sum_{n=0}^{m-1}i_n\,i_{n+m}\right)}
	\omega^{-\left(\prod_{n=m}^{2m-1}i_n + \sum_{n=0}^{m-1} \left((i_n+s_n) i_{n+m}\right)\right)}\\
&= \sum_{i=0}^{d-1}\sum_{i_1=0}^{d-1}\cdots\sum_{i_{2m-1}=0}^{d-1}  \omega^{-\sum_{n=0}^{m-1}s_n\,i_{n+m}}\\
&= \sum_{i=0}^{d-1}\sum_{i_1=0}^{d-1}\cdots\sum_{i_{2m-1}=0}^{d-1} \left(\prod_{n=0}^{m-1} \omega^{-s_n\,i_{n+m}}\right)\\
&= \left(\sum_{i_m=0}^{d-1}\omega^{-s_0\,i_m}\right)\left(\sum_{i_{m+1}=0}^{d-1}\omega^{-s_1\,i_{m+1}}\right)\times\cdots\times \left(\sum_{i_{2m-1}=0}^{d-1}\omega^{-s_{m-1}\,i_{2m-1}}\right)
\end{align*}
Each of the above terms are Gaussian sums. At least one of the sums is zero, as we are computing off-peak autocorrelations. So 
\textbf{S} satisfies the second condition of the GAOP. Thus, \textbf{S} is a perfect array.
\end{proof}

In one dimension, the largest known sequence with the AOP is a \hyperlink{FrankConst}{Frank--Heimiller sequence}. The obvious 
question remains: are there GAOP--type constructions which build larger arrays than Construction \Roman{gfrank}? The following
construction produces arrays which are 4 times the size of Construction \Roman{gfrank}, but is restricted to two dimensions.\\

\const{%
\constlabel{const2rr2D}\hypertarget{const2rr2Dtarget}
Let 
$$\textbf{S}' = [{S'}_{i_0,i_1,i_2,i_3}] = \omega^{\displaystyle\left\lfloor\frac{(d\,i_0 + i_2)(d\,i_1+i_3)}{2d}\right\rfloor}$$
be a 4--dimensional array of size $2d\times 2d \times d \times d$, where $\omega = e^{2 \pi \sqrt{-1}/d}$ and $d$ is even. Let 
$$\textbf{S} = [S_{i, j}] = \omega^{\displaystyle\left\lfloor\frac{i\,j}{2d}\right\rfloor}$$ be an array of size $2d^2 \times 2d^2$ 
formed by concatenating the array $\textbf{S}'$.}

\begin{theorem}\label{const2rr2Dthm}
The array \textbf{S} from Construction \Roman{const2rr2D} is perfect. 
\end{theorem}

\begin{proof}
We show the array \textbf{S} is perfect by showing it has the GAOP for the divisor $d$. 
Firstly, we show that all distinct 2--dimensional arrays of $\textbf{S}'$ are orthogonal.
\begin{align*}
&\theta_{\textbf{S}'[i_2,i_3],\textbf{S}'[i'_2,i'_3]}(h,v) =\sum_{i_0=0}^{2d-1}\sum_{i_1=0}^{2d-1} {S'}_{i_0,i_1,i_2,i_3} {S'}_{i_0+v,i_1+h,i_2,i_3}^*\\
&= \omega^{-dhv/2} \sum_{i_0=0}^{2d-1}\sum_{i_1=0}^{2d-1} \omega^{-\frac{d}{2}h i_0 - \frac{d}{2}v i_1 + 
	\left\lfloor\frac{1}{2}(i_0 i_3 + i_1 i_2) + \frac{1}{2d}(i_2 i_3)\right\rfloor - 
	\left\lfloor \frac{1}{2}(i'_3i_0+ i'_2 i_1 + i'_3v +i'_2h) + \frac{1}{2d}(i'_2 i'_3)\right\rfloor}
\end{align*}
We split the summation into $i_0,i_1$ even and odd. Consider the case when $i_0,i_1$ are even, then we have
\begin{align*}
&\omega^{-dhv/2} \sum_{i_0=0}^{d-1}\sum_{i_1=0}^{d-1} \omega^{-\frac{d}{2}h (2 i_0) - \frac{d}{2}v (2 i_1) + 
	\lfloor\mathcal{A}\rfloor - \lfloor\mathcal{B}\rfloor}\\
&=\omega^{-dhv/2 + \lfloor\mathcal{C}\rfloor - \lfloor\mathcal{D}\rfloor} 
	\sum_{i_0=0}^{d-1}\sum_{i_1=0}^{d-1} \omega^{(i_3-i'_3-d h) i_0 + (i_2-i'_2-d v)i_1}\\	
&=\omega^{-dhv/2 + \lfloor\mathcal{C}\rfloor -  \lfloor\mathcal{D}\rfloor} 
	\left(\sum_{i_0=0}^{d-1} \omega^{(i_3-i'_3-d h) i_0}\right)
	\left(\sum_{i_1=0}^{d-1} \omega^{(i_2-i'_2-d v)i_1}\right),	
\end{align*}

where $\mathcal{A} = \frac{1}{2}(2 i_0 i_3 + 2 i_1 i_2) + \frac{1}{2d}(i_2i_3)$, 
	$\mathcal{B} =  \frac{1}{2}(2 i_0 i'_3 + vi'_3 + 2 i_1 i'_2 + hi'_2) + \frac{1}{2d}(i'_2 i'_3)$, 
	$\mathcal{C} = \frac{1}{2d}(i_2i_3)$, and 
	$\mathcal{D} =   \frac{1}{2}(vi'_3+hi'_2) + \frac{1}{2d}(i'_2 i'_3)$. Both 
the Gaussian sums above are zero as $i_3 \neq i'_3$ and $i_2\neq i'_2$. Now consider the case when $i_0,i_1$ are odd, then we have
\begin{align*}
&\omega^{-dhv/2} \sum_{i_0=0}^{d-1}\sum_{i_1=0}^{d-1} \omega^{-\frac{d}{2}h (2 i_0+1) - \frac{d}{2}v (2 i_1+1) + 
	\left\lfloor \mathcal{A}\right\rfloor - \left\lfloor \mathcal{B}\right\rfloor}\\
&= \omega^{-\frac{d}{2}(vh+h+v) + \left\lfloor \mathcal{C}\right\rfloor - \left\lfloor \mathcal{D}\right\rfloor} \sum_{i_0=0}^{d-1}\sum_{i_1=0}^{d-1} \omega^{(i_3-i'_3 - d h)i_0 + (i_2-i'_2-dv)i_1}\\
&= \omega^{-\frac{d}{2}(vh+h+v)  + \left\lfloor \mathcal{C}\right\rfloor - \left\lfloor \mathcal{D}\right\rfloor} \left(\sum_{i_0=0}^{d-1}\omega^{(i_3-i'_3 - d h)i_0}\right)
	\left(\sum_{i_1=0}^{d-1}\omega^{(i_2-i'_2-dv)i_1}\right),
\end{align*}

where $\mathcal{A} = \frac{1}{2}((2 i_0+1) i_3 + (2i_1+1) i_2) + \frac{1}{2d}(i_2i_3)$, 
	$\mathcal{B} = \frac{1}{2}((2 i_0+1)i'_3 + vi'_3 + (2 i_1+1) i'_2 + hi'_2) + \frac{1}{2d}(i'_2 i'_3)$, 
	$\mathcal{C} = \frac{1}{2}(i_3+i_2) + \frac{1}{2d}(i_2i_3)$, and 
	$\mathcal{D} = \frac{1}{2}(i'_3+vi'_3+i'_2+hi'_2) + \frac{1}{2d}(i'_2i'_3)$. 
Both the Gaussian sums above are zero as $i_3 \neq i'_3$ and $i_2\neq i'_2$. Thus, \textbf{S}  satisfies the first condition of the GAOP.\\ 

We now show \textbf{S} satisfies the second condition of the GAOP. 

\begin{align*}
\sum_{i_0=0}^{2d-1}\sum_{i_1=0}^{2d-1} \theta_{\textbf{S}'[i_2,i_3]}(h,v) &= \sum_{i_0=0}^{2d-1}\sum_{i_1=0}^{2d-1}\sum_{i_2=0}^{d-1}\sum_{i_3=0}^{d-1}
	{S'}_{i_0,i_1,i_2,i_3} {S'}_{i_0+v,i_1+h,i_2,i_3}^*\\
&= \omega^{-\frac{d}{2}v h} \sum_{i_0=0}^{2d-1}\sum_{i_1=0}^{2d-1}\sum_{i_2=0}^{d-1}\sum_{i_3=0}^{d-1} 
	\omega^{-\frac{d}{2}h i_0 - \frac{d}{2}v i_1 + \left\lfloor\mathcal{A}\right\rfloor - \left\lfloor\mathcal{B}\right\rfloor}\\
&= \omega^{-\frac{d}{2}v h} \sum_{i_0=0}^{2d-1}\sum_{i_1=0}^{2d-1}\sum_{i_2=0}^{d-1}\sum_{i_3=0}^{d-1} 
	\omega^{-\frac{d}{2}h i_0 - \frac{d}{2}v i_1 + \left\lfloor\mathcal{A}\right\rfloor - \left\lfloor\mathcal{A} + \frac{1}{2}(h i_2 + v i_3)\right\rfloor},
\end{align*}
where $\mathcal{A} = \frac{1}{2}(i_0i_3+i_1i_2) + \frac{1}{2d}(i_2i_3)$ and $\mathcal{B} = \frac{1}{2}(i_0i_3 + i_1i_2+hi_2+vi_3) + \frac{1}{2d}(i_2i_3)$. We split the summation into $i_2,i_3$ even and odd. Consider the case when $i_2, i_3$ are 
even. Then we have
\begin{align*}
&\omega^{-\frac{d}{2}v h} \sum_{i_0=0}^{2d-1}\sum_{i_1=0}^{2d-1}\sum_{i_2=0}^{d/2-1}\sum_{i_3=0}^{d/2-1} 
	\omega^{-\frac{d}{2}h i_0 - \frac{d}{2}v i_1 + \left\lfloor\mathcal{C}\right\rfloor - \left\lfloor\mathcal{C} + \frac{1}{2}(2 h i_2 + 2 v i_3)\right\rfloor}\\
&= \omega^{-\frac{d}{2}v h} \sum_{i_0=0}^{2d-1}\sum_{i_1=0}^{2d-1}\sum_{i_2=0}^{d/2-1}\sum_{i_3=0}^{d/2-1} 
	\omega^{-\frac{d}{2}h i_0 - \frac{d}{2}v i_1 - hi_2 - vi_3}\\
&= \omega^{-\frac{d}{2}v h}\left(\sum_{i_0=0}^{2d-1} \omega^{-\frac{d}{2} h i_0}\right)\left(\sum_{i_1=0}^{2d-1} \omega^{-\frac{d}{2} v i_1}\right)
	\left(\sum_{i_2=0}^{d/2-1} \omega^{-h i_2}\right)\left(\sum_{i_3=0}^{d/2-1} \omega^{-v i_3}\right),
\end{align*}
where $\mathcal{C} = i_0i_3+i_1i_2 + \frac{2}{d}(i_2i_3)$. When $h$ is odd: $\sum_{i_0=0}^{2d-1} \omega^{-\frac{d}{2} h i_0} = 0$, otherwise 
for $h$ even: $\sum_{i_2=0}^{d/2-1} \omega^{-h i_2} = 0$ for $h\neq 0$, similarly when $v$ is odd: 
$\sum_{i_1=0}^{2d-1} \omega^{-\frac{d}{2} v i_1}$, otherwise for $v$ even: $\sum_{i_3=0}^{d/2-1} \omega^{-v i_3} = 0$ for $v\neq 0$.\\

Now consider the case when $i_2,i_3$ is odd. Then we have 
\begin{align*}
&\omega^{-\frac{d}{2}v h} \sum_{i_0=0}^{2d-1}\sum_{i_1=0}^{2d-1}\sum_{i_2=0}^{d/2-1}\sum_{i_3=0}^{d/2-1} 
	\omega^{-\frac{d}{2}h i_0 - \frac{d}{2}v i_1 + \left\lfloor\mathcal{D}\right\rfloor - \left\lfloor\mathcal{D} + \frac{1}{2}(2 h i_2 + h + 2 v i_3 + v)\right\rfloor}\\
&= \omega^{-\frac{d}{2}v h - h - v} \sum_{i_0=0}^{2d-1}\sum_{i_1=0}^{2d-1}\sum_{i_2=0}^{d/2-1}\sum_{i_3=0}^{d/2-1} 
	\omega^{-\frac{d}{2}h i_0 - \frac{d}{2}v i_1 - hi_2 - vi_3}\\
&= \omega^{-\frac{d}{2}v h - h - v}\left(\sum_{i_0=0}^{2d-1} \omega^{-\frac{d}{2} h i_0}\right)\left(\sum_{i_1=0}^{2d-1} \omega^{-\frac{d}{2} v i_1}\right)
	\left(\sum_{i_2=0}^{d/2-1} \omega^{-h i_2}\right)\left(\sum_{i_3=0}^{d/2-1} \omega^{-v i_3}\right),
\end{align*}
where $\mathcal{D} = \frac{1}{2}(i_0(2i_3+1)+i_1(2i_2+1)) + \frac{1}{2d}(2i_2+1)(2i_3+1)$. Which as before, is zero. So \textbf{S} satisfies the 
second condition of the GAOP. Thus \textbf{S} is perfect. 
\end{proof}

\bigskip

The following construction is an $n$--dimensional generalization of Construction \Roman{const2rr2D}.\\

\const{%
\constlabel{const2rr} \hypertarget{const2rrtarget}
Let $$\textbf{S}' = [{S'}_{i_0, i_1, \cdots, i_{4m-1}}] = 
\omega^{\displaystyle \left\lfloor\frac{\sum _{n=0}^{m-1} (d\, i_n + i_{n+2m}) (d\, i_{n+m} + i_{n+3m})}{2d}\right\rfloor}$$ 
be a $4m$--dimensional array of size 
$\overbrace{2d\times 2d\times \cdots \times 2d}^\text{$2m$ terms}\times\overbrace{d\times d\times\cdots\times d}^\text{$2m$ terms}$, 
where $\omega = e^{2 \pi \, \sqrt{-1}/d}$ and $d$ is even. Let $$\textbf{S} = [S_{i_0, i_1, \cdots, i_{2m-1}}] = 
\omega^{\displaystyle \left\lfloor\frac{\sum _{n=0}^{m-1} i_n\, i_{n+m}}{2d}\right\rfloor}$$ 
be a $2m$--dimensional array of size $2d^2 \times 2d^2\times\cdots\times 2d^2$ formed by concatenating the array $\textbf{S}'$.}

\bigskip

\begin{theorem}\label{newcons2}
The array \textbf{S} from Construction \Roman{const2rr} is perfect. 
\end{theorem}

\begin{proof}
We show the array \textbf{S} is perfect by showing it has the GAOP for the divisor $d$. Firstly, we show that all distinct $m$--dimensional arrays of $\textbf{S}'$ are orthogonal.
\begin{align*}
&\theta_{\textbf{S}'[i_{2m}, i_{2m+1}, \cdots, i_{4m-1}], \textbf{S}'[i'_{2m}, i'_{2m+1}, \cdots, i'_{4m-1}]}(s_0, s_1, \cdots, s_{2m-1}) = \\
&\sum_{i_0=0}^{d-1}\sum_{i_1=0}^{d-1} \cdots \sum_{i_{2m-1}=0}^{d-1} 
	{S'}_{i_0, i_1, \cdots, i_{4m-1}} \, {S'}_{i_0+s_0, i_1+s_1, \cdots, i_{2m-1}+s_{2m-1}, i'_{2m}, i'_{2m+1}, \cdots, i'_{4m-1}}^*\\
&=\sum_{i_0=0}^{d-1}\sum_{i_1=0}^{d-1} \cdots \sum_{i_{2m-1}=0}^{d-1} 	
	\left(\omega^{\left\lfloor\frac{\sum _{n=0}^{m-1} (d\, i_n + i_{n+2m}) 
	(d\, i_{n+m} + i_{n+3m})}{2d}\right\rfloor}\times\right.\\
&\qquad\qquad\left. \omega^{ -\left\lfloor\frac{
		\sum_{n=0}^{m-1} (d\, (i_n + s_n) + {i'}_{n+2m}) (d\, (i_{n+m} + s_{n+m}) + {i'}_{n+3m})}{2d}\right\rfloor}\right)\\
&=\sum_{i_0=0}^{d-1}\sum_{i_1=0}^{d-1} \cdots \sum_{i_{2m-1}=0}^{d-1}  \omega^{
	\left\lfloor\frac{\sum_{n=0}^{m-1} \mathcal{A}}{2d}\right\rfloor}
	\omega^{-\left\lfloor\frac{\sum_{n=0}^{m-1} \mathcal{B}}{2d}\right\rfloor},
\end{align*}
where $\mathcal{A} = d^2 i_n i_{m+n}+d\, i_{m+n} i_{2 m+n}+d\, i_n i_{3 m+n}+i_{2 m+n} i_{3m+n}$ and 
$\mathcal{B} = d^2 s_n i_{m+n}+d^2 i_n s_{m+n}+d^2 i_n i_{m+n}+d^2 s_n s_{m+n}+d
   {i'}_{2 m+n} i_{m+n} +d {i'}_{3 m+n} i_n+d s_n {i'}_{3 m+n}+d s_{m+n} {i'}_{2 m+n}+{i'}_{2 m+n} {i'}_{3 m+n}$. We split the 
summation above into $i_0, i_1, \cdots, i_{2m-1}$ even and odd. Consider the case when $i_0, i_1, \cdots, i_{2m-1}$ 
is even. Then we have 
\begin{align*}
&\sum_{i_0=0}^{d/2-1}\sum_{i_1=0}^{d/2-1} \cdots \sum_{i_{2m-1}=0}^{d/2-1}  \omega^{
	\sum_{n=0}^{m-1} \left(2 d i_n i_{m+n}+i_{2 m+n} i_{m+n}+i_n i_{3 m+n}\right) + 
	 \left\lfloor \frac{\sum_{n=0}^{m-1}\mathcal{C}}{2 d}\right\rfloor}   \times\\
&\qquad\qquad\omega^{-\sum_{n=0}^{m-1}\left( d s_n i_{m+n}+d i_n s_{m+n}+2 d i_n i_{m+n}+i_{m+n} {i'}_{2 m+n}+{i'}_{3
   m+n}i_n \right) - \left\lfloor \frac{\sum_{n=0}^{m-1}\mathcal{D}}{2 d}\right\rfloor}\\
&=\sum_{i_0=0}^{d/2-1}\sum_{i_1=0}^{d/2-1} \cdots \sum_{i_{2m-1}=0}^{d/2-1} \omega^{
	-\sum_{n=0}^{m-1}\left(\left(i_{3 m+n}-{i'}_{3 m+n} - d s_{m+n}\right)i_n + \left(i_{2 m+n}-{i'}_{2 m+n} - d s_n\right)i_{m+n} \right) + 
\left\lfloor \frac{\sum_{n=0}^{m-1} \mathcal{C}}{2 d}\right\rfloor - \left\lfloor \frac{\sum_{n=0}^{m-1} \mathcal{D}}{2 d}\right\rfloor}\\
&=\omega^{\left\lfloor \frac{\sum_{n=0}^{m-1}\mathcal{C}}{2 d}\right\rfloor - \left\lfloor \frac{\sum_{n=0}^{m-1}\mathcal{D}}{2 d}\right\rfloor}
\sum_{i_0=0}^{d/2-1}\sum_{i_1=0}^{d/2-1} \cdots \sum_{i_{2m-1}=0}^{d/2-1} \omega^{-\sum_{n=0}^{m-1}\left(
\left(i_{3 m+n}-{i'}_{3 m+n} - d s_{m+n}\right)i_n + \left(i_{2 m+n}-{i'}_{2 m+n} - d s_n\right)i_{m+n}\right)}\\
&= \omega^{\left\lfloor \frac{\sum_{n=0}^{m-1}\mathcal{C}}{2 d}\right\rfloor - \left\lfloor \frac{\sum_{n=0}^{m-1}\mathcal{D}}{2 d}\right\rfloor}
\left(\sum_{i_0=0}^{d/2-1} \omega^{(i_{3m} - {i'}_{3m} - d s_m) i_0} \right) 
\left(\sum_{i_1=0}^{d/2-1} \omega^{(i_{3m+1} - {i'}_{3m+1} - d s_{m+1}) i_1} \right) \times\cdots\times\\
&\qquad\qquad\left(\sum_{i_{m-1}=0}^{d/2-1} \omega^{(i_{4m-1} - {i'}_{4m-1} - d s_{2m-1}) i_{m-1}} \right)
\times \left(\sum_{i_m=0}^{d/2-1} \omega^{(i_{2m} - {i'}_{2m} - d s_0) i_m} \right)\times \\
&\qquad\qquad\left(\sum_{i_{m+1}=0}^{d/2-1} \omega^{(i_{2m+1} - {i'}_{2m+1} - d s_1) i_{m+1}} \right)\times\cdots
\times\left(\sum_{i_{2m-1}=0}^{d/2-1} \omega^{(i_{3m-1} - {i'}_{3m-1} - d s_{m-1}) i_{2m-1}} \right),
\end{align*}
where $\mathcal{C} = i_{2 m+n} i_{3 m+n}$ and $\mathcal{D} = d^2 s_n s_{m+n}+d s_n {i'}_{3 m+n}+d s_{m+n} {i'}_{2 m+n} +{i'}_{2 m+n} {i'}_{3 m+n}$. All the Gaussian sums above are zero as $i_k\neq {i'}_k$, for all $k$. Now consider the case where $i_0, i_1, \cdots, i_{2m-1}$ is odd. In 
this case we have $\mathcal{C} = d^2+d i_{2 m+n}+d i_{3 m+n}+i_{2 m+n} i_{3 m+n}$ and $\mathcal{D} = 
d^2 s_n s_{m+n}+d^2 s_{m+n}+d^2 s_n+d^2+d s_n {i'}_{3 m+n}+d {i'}_{2 m+n} s_{m+n}+d {i'}_{2 m+n}+
d {i'}_{3 m+n}+{i'}_{2 m+n} {i'}_{3 m+n}$, but the product of Gaussian sums is the same as the case above. Thus, \textbf{S} satisfies
the first condition of the GAOP. \\

We now show \textbf{S} satisfies the second condition of the GAOP.
\begin{align*}
&\sum_{i_0=0}^{d-1}\sum_{i_1=0}^{d-1}\cdots\sum_{i_{2m-1}=0}^{d-1} \theta_{\textbf{S}'[i_{2m}, i_{2m+1}, \cdots, i_{4m-1}]}
	(s_0, s_1, \cdots, s_{2m-1})=\\
&\sum_{i_0=0}^{d-1}\sum_{i_1=0}^{d-1}\cdots\sum_{i_{4m-1}=0}^{d-1} S_{i_0, i_1, \cdots, i_{4m-1}}
	S_{i_0+s_0, i_1+s_1, \cdots, i_{2m-1}+s_{2m-1}, i_{2m}, i_{2m+1}, \cdots, i_{4m-1}}^*\\
&=\sum_{i_0=0}^{d-1}\sum_{i_1=0}^{d-1}\cdots\sum_{i_{4m-1}=0}^{d-1}
	\omega^{\left\lfloor\frac{\sum _{n=0}^{m-1} (d\, i_n + i_{n+2m}) (d\, i_{n+m} + i_{n+3m})}{2d}\right\rfloor} \times\\
&\qquad\qquad \omega^{ -\left\lfloor\frac{\sum _{n=0}^{m-1} (d\, (i_n + s_n) + i_{n+2m}) (d\, (i_{n+m}+s_{n+m}) + i_{n+3m})}{2d}\right\rfloor}.
\end{align*}
We split the 
summation above into $i_0, i_1, \cdots, i_{4m-1}$ even and odd. Consider the case when $i_0, i_1, \cdots, i_{4m-1}$ 
is even. Then we have 
$$\sum_{i_0=0}^{d/2-1}\sum_{i_1=0}^{d/2-1}\cdots\sum_{i_{4m-1}=0}^{d/2-1} \omega^{-\sum_{n=0}^{m-1}\left(
	d s_{m+n} i_n + d s_n i_{m+n} + s_{m+n}i_{2 m+n} + s_n i_{3 m+n}\right) + 
	\left\lfloor \mathcal{A}\right\rfloor - 
	\left\lfloor \mathcal{A} + \frac{1}{2}\sum _{n=0}^{m-1}d s_n s_{m+n}\right\rfloor},$$
where $\mathcal{A} = \frac{1}{d}\sum _{n=0}^{m-1}2 i_{2 m+n} i_{3 m+n}$. As $d$ is even, 
$\left\lfloor \mathcal{A} + \frac{1}{2}\sum _{n=0}^{m-1}d s_n s_{m+n}\right\rfloor = 
\left\lfloor \mathcal{A}\right\rfloor +  \frac{1}{2}\sum _{n=0}^{m-1}d s_n s_{m+n}$, then the summation above becomes
\begin{align*}
&\omega^{-\frac{1}{2}\sum _{n=0}^{m-1}d s_n s_{m+n}}
\sum_{i_0=0}^{d/2-1}\sum_{i_1=0}^{d/2-1}\cdots\sum_{i_{4m-1}=0}^{d/2-1} \omega^{-\sum_{n=0}^{m-1}\left(
	d s_{m+n} i_n + d s_n i_{m+n} + s_{m+n}i_{2 m+n} + s_n i_{3 m+n}\right)}\\
&=\omega^{-\frac{1}{2}\sum _{n=0}^{m-1}d s_n s_{m+n}}
\sum_{i_0=0}^{d/2-1}\sum_{i_1=0}^{d/2-1}\cdots\sum_{i_{4m-1}=0}^{d/2-1} \omega^{-\sum_{n=0}^{m-1}\left(
	d s_{m+n} i_n + d s_n i_{m+n}\right)}
	\omega^{-\sum_{n=0}^{m-1}\left(s_{m+n}i_{2 m+n} + s_n i_{3 m+n}\right)}\\
&=\omega^{-\frac{1}{2}\sum _{n=0}^{m-1}d s_n s_{m+n}}
\left(\sum_{i_0=0}^{d/2-1}\sum_{i_1=0}^{d/2-1}\cdots\sum_{i_{2m-1}=0}^{d/2-1} \omega^{-\sum_{n=0}^{m-1}\left(
	d s_{m+n} i_n + d s_n i_{m+n}\right)}\right)\times\\
&\qquad\qquad\left(\sum_{i_{2m}=0}^{d/2-1}\sum_{i_{2m+1}=0}^{d/2-1}\cdots\sum_{i_{4m-1}=0}^{d/2-1}
	\omega^{-\sum_{n=0}^{m-1}\left(s_{m+n}i_{2 m+n} + s_n i_{3 m+n}\right)}\right)\\
&=\omega^{-\frac{1}{2}\sum _{n=0}^{m-1}d s_n s_{m+n}}
\left(\sum_{i_0=0}^{d/2-1}\sum_{i_1=0}^{d/2-1}\cdots\sum_{i_{2m-1}=0}^{d/2-1} \omega^{-\sum_{n=0}^{m-1}\left(
	d s_{m+n} i_n + d s_n i_{m+n}\right)}\right)\times\\
&\qquad\qquad\left(\sum_{i_{2m}=0}^{d/2-1}\sum_{i_{2m+1}=0}^{d/2-1}\cdots\sum_{i_{3m-1}=0}^{d/2-1}
	\omega^{-\sum_{n=0}^{m-1}s_{m+n}i_{2 m+n}}\right)
	\left(\sum_{i_{3m}=0}^{d/2-1}\sum_{i_{3m+1}=0}^{d/2-1}\cdots\sum_{i_{4m-1}=0}^{d/2-1}
	\omega^{-\sum_{n=0}^{m-1}s_n i_{3 m+n}}\right)\\
&=\omega^{-\frac{1}{2}\sum _{n=0}^{m-1}d s_n s_{m+n}}
\left(\sum_{i_0=0}^{d/2-1}\sum_{i_1=0}^{d/2-1}\cdots\sum_{i_{2m-1}=0}^{d/2-1} \omega^{-\sum_{n=0}^{m-1}\left(
	d s_{m+n} i_n + d s_n i_{m+n}\right)}\right)\times\\
&\qquad\qquad\left( \sum_{i_{2m}=0}^{d/2-1} \omega^{-s_m i_{2m}} \right)
\left( \sum_{i_{2m+1}=0}^{d/2-1} \omega^{-s_{m+1} i_{2m+1}} \right)\times\cdots\times
\left( \sum_{i_{3m-1}=0}^{d/2-1} \omega^{-s_{2m-1} i_{3m-1}} \right)\times\\
&\qquad\left( \sum_{i_{3m}=0}^{d/2-1} \omega^{-s_0 i_{3m}} \right)
\left( \sum_{i_{3m+1}=0}^{d/2-1} \omega^{-s_1 i_{3m+1}} \right)\times\cdots\times
\left( \sum_{i_{4m-1}=0}^{d/2-1} \omega^{-s_{m-1} i_{4m-1}} \right).
\end{align*}
As we are computing off-peak autocorrelations, at least one of the Gaussian sums above is zero. A similar calculation
shows the summation is zero for $i_0, i_1, \cdots, i_{4m-1}$ odd. So, \textbf{S} satisfies the second condition of the 
GAOP. Thus, \textbf{S} is perfect.
\end{proof}

It is currently unknown if constructions similar to Construction \Roman{const2rr2D} and Construction \Roman{const2rr} exist for 
$d$ odd. The arrays $\textbf{S}'$ and \textbf{S} in Constructions \Roman{const2rr2D} and \Roman{const2rr} are not perfect for $d$ odd. 

\bigskip

The following construction is a multi--dimensional generalization of \hyperlink{MilewskiConst}{Milewski's sequence construction}. \\

\const{%
\constlabel{constmilewg}\hypertarget{constmilewgconst}
Let $\textbf{u} = [u_0,u_1, \cdots, u_{r-1}]$ be a Chu sequence \cite{Chu1972}, then let
$$\textbf{S}' = [{S'}_{i_0,i_1,\cdots, i_{2m-1}}] = \left(\prod_{n=0}^{m-1}u_{i_n}\right) \omega^{\displaystyle
	\left(\prod_{n=m}^{2 m-1} i_n + \sum_{n=0}^{m-1} i_n\, i_{n+m}\right)}$$
be a $2m$--dimensional array of size $\overbrace{r^{k+1} \times\cdots\times r^{k+1}}^{\text{$m$ terms}} 
	\times \overbrace{r^k \times\cdots\times r^k}^{\text{$m$ terms}}$,
where $\omega = e^{2 \pi \sqrt{-1}/r^{k+1}}$ and $r$ is even. Let \textbf{S} be the $m$--dimensional array of size 
$r^{2k+1}\times r^{2k+1}\times\cdots\times r^{2k+1}$ formed by concatenating the array $\textbf{S}'$.}

\bigskip

\begin{theorem}\label{newcons2}
The array \textbf{S} from Construction \Roman{constmilewg} is perfect. 
\end{theorem}

\begin{proof}
We show the array \textbf{S} is perfect by showing it has the GAOP for the divisor $r^k$. Firstly, we show that all 
distinct $m$--dimensional arrays of $\textbf{S}'$ are orthogonal.
\begin{align*}
&\theta_{\textbf{S}'[i_m,i_{m+1}, \cdots, i_{2m-1}],\textbf{S}'[i'_m,i'_{m+1}, \cdots, i'_{2m-1}]}(s_0,s_1, \cdots, s_{m-1})=\\
&\sum_{i_0=0}^{r^{k+1}-1} \sum_{i_1=0}^{r^{k+1}-1} \cdots \sum_{i_{m-1}=0}^{r^{k+1}-1}
	{S'}_{i_0, i_1, \cdots, i_{2m-1}} {S'}_{i_0+s_0, i_1+s_1, \cdots, i_{m-1} + s_{m-1}, i'_{m}, i'_{m+1}, \cdots, i'_{2m-1}}^*\\
&= \sum_{i_0=0}^{r^{k+1}-1} \sum_{i_1=0}^{r^{k+1}-1} \cdots \sum_{i_{m-1}=0}^{r^{k+1}-1}
	\left(\left(\prod_{n=0}^{m-1} u_{i_n}\right) \omega^{\prod_{n=m}^{2m-1} i_n + \sum_{n=0}^{m-1} i_n i_{n+m}}\right)
		\left(\left(\prod_{n=0}^{m-1} u_{i_n+s_n}\right)\omega^{\prod_{n=m}^{2m-1} {i'}_n + \sum_{n=0}^{m-1} (i_n+s_n) i'_{n+m}}\right)^*\\
&= \omega^{\prod_{n=m}^{2m-1} {i_n} - \prod_{n=m}^{2m-1} {i'}_n - \sum_{n=0}^{m-1}s_n i_{n+m}}
		\sum_{i_0=0}^{r^{k+1}-1} \sum_{i_1=0}^{r^{k+1}-1} \cdots \sum_{i_{m-1}=0}^{r^{k+1}-1}
			\left(\left(\prod_{n=0}^{m-1}u_{i_n} u_{i_n+s_n}^*\right)\omega^{-\sum_{n=0}^{m-1} (i_{n+m} - 'i_{n+m})i_n}\right)\\
&= \omega^{\prod_{n=m}^{2m-1} {i_n} - \prod_{n=m}^{2m-1} {i'}_n - \sum_{n=0}^{m-1}s_n i_{n+m}}
		\left(\sum_{i_0=0}^{r^{k+1}-1} \sum_{i_1=0}^{r^{k+1}-1} \cdots \sum_{i_{m-1}=0}^{r^{k+1}-1}
			\prod_{n=0}^{m-1}u_{i_n} u_{i_n+s_n}^* \omega^{(i'_{n+m} - i_{n+m})i_n} \right)\\
&= \omega^{\prod_{n=m}^{2m-1} {i_n} - \prod_{n=m}^{2m-1} {i'}_n - \sum_{n=0}^{m-1}s_n i_{n+m}}
		\left(\prod_{n=0}^{m-1}\sum_{i_n=0}^{r^{k+1}-1}u_{i_n} u_{i_n+s_n}^* \omega^{(i'_{n+m} - i_{n+m})i_n}\right)
\end{align*}
As $r$ is even, a term in the Chu sequence is given by $u_n = e^{\left(\frac{\pi \sqrt{-1}}{r}\right)p n^2}$, where $p$
is relatively prime to $r^{k+1}$.  Then we have: 
$\sum_{i_n=0}^{r^{k+1}-1}u_{i_n} u_{i_n+s_n}^* \omega^{(i'_{n+m} - i_{n+m})i_n} = 
\sum_{i_n=0}^{r^{k+1}-1} e^{\left(\frac{2 \pi \sqrt{-1}}{r^{k+1}}\right)\left(\frac{r^k}{2}p i_n^2 - 
\frac{r^k}{2}p (i_n+s_n)^2 +  i'_{n+m} - i_{n+m}\right)i_n} = 
e^{\left(\frac{2 \pi \sqrt{-1}}{r^{k+1}}\right)\left(-\frac{1}{2}p r^k s_n^2\right)} 
\sum_{i_n=0}^{r^{k+1}-1} e^{\left(\frac{2 \pi \sqrt{-1}}{r^{k+1}}\right)\left(-r^k p s_n + i'_{n+m} - i_{n+m}\right)i_n}
= 0$ as $i'_{n+m} \neq i_{n+m}$. So \textbf{S} satisfies the first condition of the GAOP. We now show \textbf{S} 
satisfies the second condition of the GAOP.

\begin{align*}
&\sum_{i_0 = 0}^{r^{k+1}-1}\cdots\sum_{i_{m-1} = 0}^{r^{k+1}-1}
	\theta_{\textbf{S}'[i_m,i_{m+1}, \cdots, i_{2m-1}]}(s_0,s_1,\cdots,s_{m-1})=\\
&\sum_{i_0 = 0}^{r^{k+1}-1}\cdots\sum_{i_{m-1} = 0}^{r^{k+1}-1}
	\sum_{i_{m} = 0}^{r^{k}-1}\cdots\sum_{i_{2m-1} = 0}^{r^{k}-1}
	{S'}_{i_0,i_1,\cdots, i_{2m-1}} {S'}_{i_0+s_0, i_1+s_1, \cdots, i_{m-1}+s_{m-1}, i_m,i_{m+1}, \cdots, i_{2m-1}}^*\\
&= \sum_{i_0 = 0}^{r^{k+1}-1}\cdots\sum_{i_{m-1} = 0}^{r^{k+1}-1}
	\sum_{i_{m} = 0}^{r^{k}-1}\cdots\sum_{i_{2m-1} = 0}^{r^{k}-1}
	\left(\left(\prod_{n=0}^{m-1}u_{i_n}\right) \omega^{\left(\prod_{n=m}^{2 m-1} i_n + \sum_{n=0}^{m-1} i_n\, i_{n+m}\right)}	\right)\times\\
&\qquad\qquad\left(\left(\prod_{n=0}^{m-1}u_{i_n+s_n}\right) \omega^{\left(\prod_{n=m}^{2 m-1} i_n + \sum_{n=0}^{m-1} (i_n+s_n)\, i_{n+m}\right)}	\right)^*\\
&= \sum_{i_0 = 0}^{r^{k+1}-1}\cdots\sum_{i_{m-1} = 0}^{r^{k+1}-1}
	\sum_{i_{m} = 0}^{r^{k}-1}\cdots\sum_{i_{2m-1} = 0}^{r^{k}-1}
	\left(\left(\prod_{n=0}^{m-1}u_{i_n}u_{i_n+s_n}^*\right)\omega^{-\sum_{n=0}^{m-1}s_n i_{n+m}} \right)\\
&= \sum_{i_0 = 0}^{r^{k+1}-1}\cdots\sum_{i_{m-1} = 0}^{r^{k+1}-1}
	\sum_{i_{m} = 0}^{r^{k}-1}\cdots\sum_{i_{2m-1} = 0}^{r^{k}-1}
	\left(\left(\prod_{n=0}^{m-1}u_{i_n}u_{i_n+s_n}^*\right) \left(\prod_{n=0}^{m-1}\omega^{-s_n i_{n+m}} \right)\right)\\
&= \left(\sum_{i_0 = 0}^{r^{k+1}-1}\cdots\sum_{i_{m-1} = 0}^{r^{k+1}-1}
	\left(\prod_{n=0}^{m-1}u_{i_n}u_{i_n+s_n}^*\right)\right)
	\left(\sum_{i_{m} = 0}^{r^{k}-1}\cdots\sum_{i_{2m-1} = 0}^{r^{k}-1}
	\left(\prod_{n=0}^{m-1}\omega^{-s_n i_{n+m}} \right)\right)\\
&= \left(\prod_{n=0}^{m-1}\sum_{i_n = 0}^{r^{k}-1} u_{i_n}u_{i_n+s_n}^*\right)
	\left(\prod_{n=0}^{m-1}\sum_{i_n = 0}^{r^{k}-1}\omega^{-s_n i_{n+m}}\right)
\end{align*}
For $s_n \neq 0 \mod r$, $\sum_{i_n = 0}^{r^{k}-1} u_{i_n}u_{i_n+s_n}^* = 0$, otherwise for $s_n = 0 \mod r$ and $s_n \neq 0$,
$\sum_{i_n = 0}^{r^{k}-1}\omega^{-s_n i_{n+m}} = 0$. So \textbf{S} satisfies the second condition of the GAOP. Thus \textbf{S} is 
perfect.
\end{proof}

Finally, we note that the array $\textbf{S}\,'$ from Construction \Roman{constmilewg} is perfect for odd $r$.

\begin{theorem}\label{newcons2thm}
The array $\textbf{S}\,'$ from Construction \Roman{constmilewg} is perfect for odd $r$. 
\end{theorem}

\begin{proof}
The proof is similar to those given for the previous constructions. 
\end{proof}

\section{Conclusions}
We have generalised the AOP to higher dimensions and showed that a $n$-dimensional
array with the GAOP is perfect. Using the GAOP, we have derived a number of perfect array
constructions. Each of these array constructions are bounded in size. It is unknown if there 
exist array constructions with the GAOP which are unbounded in size. 

\bibliographystyle{abbrv}

\begin{thebibliography}{99} 
\bibitem[Blake, 2013]{Blake2012} S. Blake, T. E. Hall, A. Z. Tirkel, ``Arrays over Roots of Unity with Perfect Autocorrelation and Good ZCZ Cross--Correlation", \textit{Advances in Mathematics of Communications (AMC)}, vol. 7, no. 3, pp. 231-242, 2013

\bibitem[Blake, 2014]{Blake2014} S. Blake, A. Z. Tirkel, ``A Construction for Perfect Autocorrelation Sequences over Roots 
of Unity", \textit{SETA 2014}, Springer LNCS 8865, pp. 104-108, November 2014

\bibitem[Chu, 1972]{Chu1972} D. C. Chu, ``Polyphase Codes With Good Periodic Correlation Properties", 
\textit{IEEE Trans. on Inform. Theory}, vol. 18, no. 4, pp. 531--532, July 1972

\bibitem[Frank, 1962]{Frank1962}R. L. Frank, S. A. Zadoff and R. Heimiller, ``Phase Shift Pulse Codes with Good Periodic Correlation Properties", \textit{IRE Transactions on Information Theory}, vol. 8, no. 6, pp. 381-382, October 1961

\bibitem[Heimiller, 1961]{Heimiller1961} R. C. Heimiller, ``Phase Shift Pulse Codes with Good Periodic Correlation Properties", 
			\textit{IRE Transactions on Information Theory}, vol. 7, no. 4, pp. 254-257, October 1961

\bibitem[Liu, 2004]{Liu2004} Y. Liu and P. Fan, ``Modified Chu sequences with smaller alphabet size", \textit{Electronics Letters}, vol. 40, no. 10, May 2004
			
\bibitem[Milewski, 1983]{Milewski1983} A. Milewski, ``Periodic Sequences with Optimal Properties for Channel Estimation and 
Fast Start-Up Equalization", \textit{IBM Journal of Research and Development}, vol. 27, no. 5, pp. 426-431, September 1983

\bibitem[Mow, 1993]{Mow1993} W. H. Mow, ``A Study of Correlation of Sequences", PhD, Department of Information Engineering, 
The Chinese University of Hong Kong, 1993

\end{thebibliography}

\section*{Appendix I -- Implementation of the Constructions}
In this appendix we show the implementations of the $n$--dimensional constructions in the computer algebra system, 
\href{http://www.wolfram.com}{Mathematica} (version 8.0). (All arrays are given in index notation, that is, the 
mapping: $e^{2 \pi \sqrt{-1} s_n/r} \rightarrow s_n$.) \\

We begin with the code for periodic cross-correlation, {\tt XCV} and autocorrelation, {\tt ACV}:
{\small
\begin{verbatim}
In[1]:= XCV[a_, b_, r_Integer] := Block[{A, B},
    A = Developer`ToPackedArray[ Exp[(2. Pi I a)/r] ];
    B = Developer`ToPackedArray[ Exp[(-2. Pi I b)/r] ]; 
    Chop[ListCorrelate[A, B, 1], 1*^-5]]
  
In[2]:= ACV[m_, r_Integer] := XCV[m, m, r]
\end{verbatim}}

The function {\tt index} takes as input the dimension, {\tt d}, and returns the index function for the $d$--dimensional array, 
$\displaystyle\prod_{n=d/2}^{d-1} i_n + \sum_{n=0}^{d/2-1} i_n\, i_{n+d/2}$ (note that in Mathematica, array indexing starts at 1):

{\small
\begin{verbatim}
In[1]:= index[d_?EvenQ] := 
     Function @@ {Sum[Slot[n] Slot[n + d/2], {n, d/2}] + Product[Slot[n], {n, d/2 + 1, d}]}
\end{verbatim}
}

For example, we compute the index function for the 2--dimensional array from Construction \Roman{gfrank}:
{\small
\begin{verbatim}
In[2]:= index[2]
Out[2]= #1 + #1 #2 &
\end{verbatim}
}

And now the index function for the 8--dimensional array: 
{\small
\begin{verbatim}
In[3]:= index[8]
Out[3]= #1 #5 + #2 #6 + #3 #7 + #4 #8 + #4 #5 #6 #7 #8 &
\end{verbatim}
}
The following function implements \hyperlink{gfrankconst}{Construction \Roman{gfrank}}. It takes as inputs the number of 
roots of unity, {\tt nr}, and the number of dimensions, {\tt nd} and returns the multi--dimensional perfect array, \textbf{S}: 
{\small
\begin{verbatim}
In[4]:= ConstructionVI[nr_Integer, nd_?EvenQ] := With[{indexF = index[2 nd]}, 
            ArrayFlatten[Mod[Array[indexF, Table[nr, {2 nd}]], nr], nd]]
\end{verbatim}
}
For example, the following is a perfect 4--dimensional binary array:
{\small
\begin{verbatim}
In[5]:= ConstructionVI[2, 4]
\end{verbatim}
${\verb~Out[5]= ~}\left[
\begin{array}{cccc}
 \left[
\begin{array}{cccc}
 1 & 1 & 0 & 1 \\
 1 & 0 & 0 & 0 \\
 0 & 0 & 1 & 0 \\
 1 & 0 & 0 & 0
\end{array}
\right] & \left[
\begin{array}{cccc}
 1 & 0 & 0 & 0 \\
 0 & 1 & 1 & 1 \\
 0 & 1 & 1 & 1 \\
 0 & 1 & 1 & 1
\end{array}
\right] & \left[
\begin{array}{cccc}
 0 & 0 & 1 & 0 \\
 0 & 1 & 1 & 1 \\
 1 & 1 & 0 & 1 \\
 0 & 1 & 1 & 1
\end{array}
\right] & \left[
\begin{array}{cccc}
 1 & 0 & 0 & 0 \\
 0 & 1 & 1 & 1 \\
 0 & 1 & 1 & 1 \\
 0 & 1 & 1 & 1
\end{array}
\right] \\
 \left[
\begin{array}{cccc}
 1 & 0 & 0 & 0 \\
 0 & 1 & 1 & 1 \\
 0 & 1 & 1 & 1 \\
 0 & 1 & 1 & 1
\end{array}
\right] & \left[
\begin{array}{cccc}
 0 & 1 & 1 & 1 \\
 1 & 0 & 0 & 0 \\
 1 & 0 & 0 & 0 \\
 1 & 0 & 0 & 0
\end{array}
\right] & \left[
\begin{array}{cccc}
 0 & 1 & 1 & 1 \\
 1 & 0 & 0 & 0 \\
 1 & 0 & 0 & 0 \\
 1 & 0 & 0 & 0
\end{array}
\right] & \left[
\begin{array}{cccc}
 0 & 1 & 1 & 1 \\
 1 & 0 & 0 & 0 \\
 1 & 0 & 0 & 0 \\
 1 & 0 & 0 & 0
\end{array}
\right] \\
 \left[
\begin{array}{cccc}
 0 & 0 & 1 & 0 \\
 0 & 1 & 1 & 1 \\
 1 & 1 & 0 & 1 \\
 0 & 1 & 1 & 1
\end{array}
\right] & \left[
\begin{array}{cccc}
 0 & 1 & 1 & 1 \\
 1 & 0 & 0 & 0 \\
 1 & 0 & 0 & 0 \\
 1 & 0 & 0 & 0
\end{array}
\right] & \left[
\begin{array}{cccc}
 1 & 1 & 0 & 1 \\
 1 & 0 & 0 & 0 \\
 0 & 0 & 1 & 0 \\
 1 & 0 & 0 & 0
\end{array}
\right] & \left[
\begin{array}{cccc}
 0 & 1 & 1 & 1 \\
 1 & 0 & 0 & 0 \\
 1 & 0 & 0 & 0 \\
 1 & 0 & 0 & 0
\end{array}
\right] \\
 \left[
\begin{array}{cccc}
 1 & 0 & 0 & 0 \\
 0 & 1 & 1 & 1 \\
 0 & 1 & 1 & 1 \\
 0 & 1 & 1 & 1
\end{array}
\right] & \left[
\begin{array}{cccc}
 0 & 1 & 1 & 1 \\
 1 & 0 & 0 & 0 \\
 1 & 0 & 0 & 0 \\
 1 & 0 & 0 & 0
\end{array}
\right] & \left[
\begin{array}{cccc}
 0 & 1 & 1 & 1 \\
 1 & 0 & 0 & 0 \\
 1 & 0 & 0 & 0 \\
 1 & 0 & 0 & 0
\end{array}
\right] & \left[
\begin{array}{cccc}
 0 & 1 & 1 & 1 \\
 1 & 0 & 0 & 0 \\
 1 & 0 & 0 & 0 \\
 1 & 0 & 0 & 0
\end{array}
\right]
\end{array}
\right]$}

\bigskip

The following code snippet gives the full calculation from Example \ref{GAOPexample}. We begin by
showing \textbf{S} satisfies the first condition of the GAOP. The following code generates all 
arrays $\textbf{S}'[i,j]$, for $0\leq i < 3$ and $0 \leq j < 3$:
{\small
\begin{verbatim}
In[6]:= allSps = Join @@ Table[
    Table[
         ConstructionVI[3, 2][[3 n + 1 + q, 3 m + 1 + r]], 
    {n, 0, 2}, {m, 0, 2}],
{q, 0, 2}, {r, 0, 2}];
\end{verbatim}}
We now generate all 36 distinct pairs of arrays, compute their cross correlation, and count the 
number of zeros in the resulting array of cross-correlation values: 
{\small
\begin{verbatim} 
In[7]:= Count[XCV[#1, #2, 3], 0, {2}] & @@@ 
      Union[ Select[Sort /@ Tuples[allSps, 2], First[#] != Last[#] &] ]
      
Out[7]= {9, 9, 9, 9, 9, 9, 9, 9, 9, 9, 9, 9, 9, 9, 9, 9, 9, 9, \
      9, 9, 9, 9, 9, 9, 9, 9, 9, 9, 9, 9, 9, 9, 9, 9, 9, 9}
\end{verbatim}}
We see that for each pair of arrays there are 9 zero cross-correlation values. Thus, each pair of 
arrays are orthogonal. We now show \textbf{S} satisfies the second condition of the GAOP. For
each array, we compute its autocorrelation for all shifts. We then sum all the autocorrelations 
together. 
{\small
\begin{verbatim}
In[8]:= Chop @ Total[ACV[#, 3] & /@ allSps]

Out[8]= {{81, 0, 0}, {0, 0, 0}, {0, 0, 0}}
\end{verbatim}}
So for all off-peak shifts, the sum of the autocorrelation is zero. Thus, \textbf{S} satisfies the 
second condition of the GAOP. 

\bigskip
%% Note: Include larger example in my thesis. 

The following function implements \hyperlink{const2rr2Dtarget}{Construction \Roman{const2rr2D}}. It takes
as input the number of roots, {\tt d}, and returns the perfect array \textbf{S}:
{\small
\begin{verbatim}
In[9]:= ConstructionVII[d_?EvenQ] := Mod[Array[Floor[#1 #2/(2 d)] &, {2 d^2, 2 d^2}], d]
\end{verbatim}
}
For example, the following is a perfect binary array: 
{\small
\begin{verbatim}
In[10]:= ConstructionVII[2]
\end{verbatim}
$\verb~Out[10]= ~\left[
\begin{array}{cccccccc}
 0 & 0 & 0 & 1 & 1 & 1 & 1 & 0 \\
 0 & 1 & 1 & 0 & 0 & 1 & 1 & 0 \\
 0 & 1 & 0 & 1 & 1 & 0 & 1 & 0 \\
 1 & 0 & 1 & 0 & 1 & 0 & 1 & 0 \\
 1 & 0 & 1 & 1 & 0 & 1 & 0 & 0 \\
 1 & 1 & 0 & 0 & 1 & 1 & 0 & 0 \\
 1 & 1 & 1 & 1 & 0 & 0 & 0 & 0 \\
 0 & 0 & 0 & 0 & 0 & 0 & 0 & 0
\end{array}
\right]$}\\

These arrays are highly symmetric and exhibit a beautiful structure. The following is a perfect 
$968\times 968$ perfect array over 22 roots of unity:
{\small
\begin{verbatim}
In[11]:= ArrayPlot[ConstructionVII[22], Frame -> False]
\end{verbatim}
\verb~Out[11]= ~ \vcenteredinclude{width=8cm}{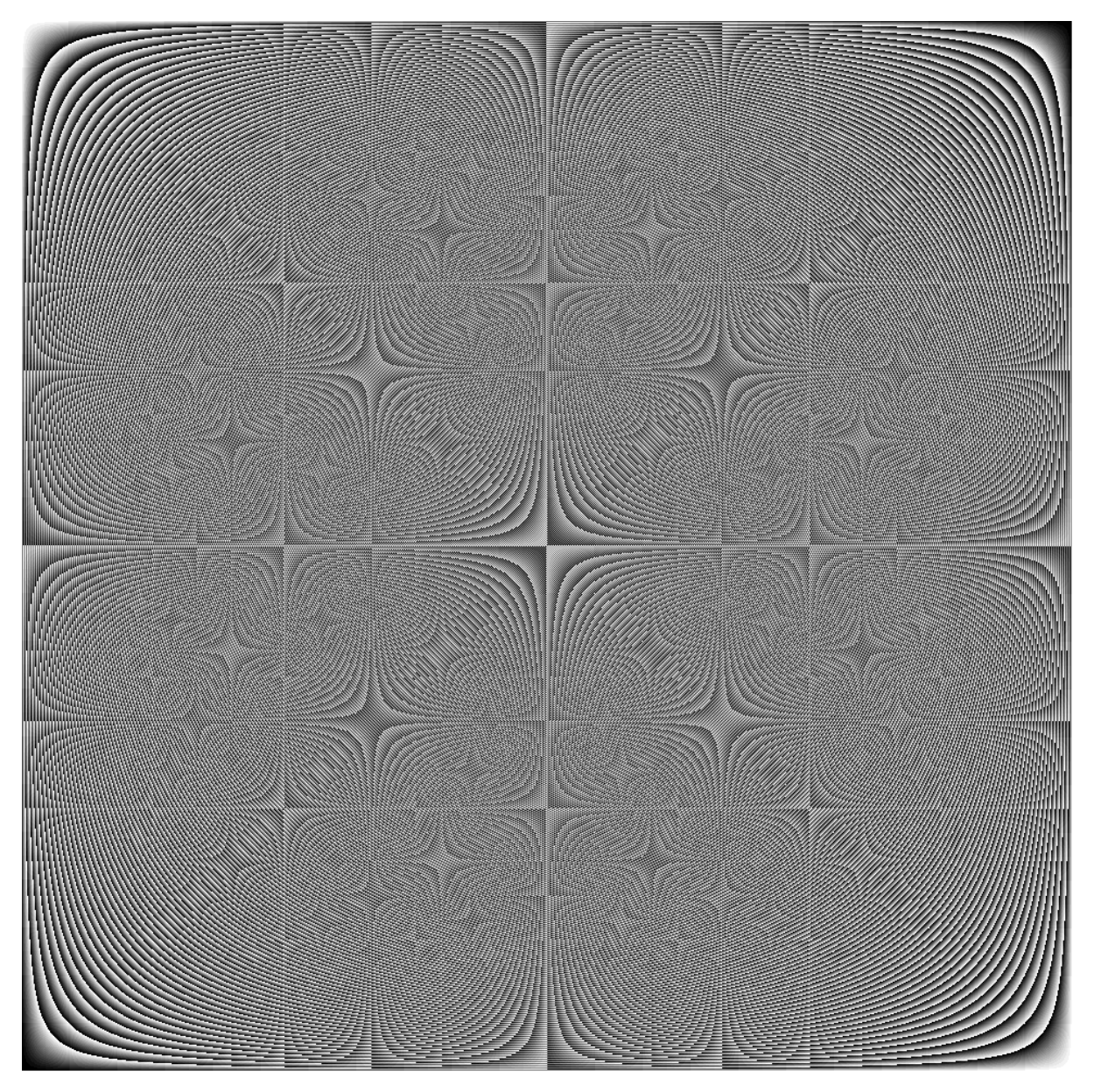}}\\

The following function implements \hyperlink{const2rrtarget}{Construction \Roman{const2rr}}. It takes as inputs the number of 
roots of unity, {\tt nr}, and the number of dimensions, {\tt nd} and returns the multi--dimensional perfect array, \textbf{S}: 
{\small
\begin{verbatim}
In[12]:= indexVIII[d_?EvenQ] := Function @@ {Sum[Slot[n] Slot[n + d/2], {n, d/2}]}
 
In[13]:= ConstructionVIII[nr_Integer, nd_?EvenQ] := With[{indexP = indexVIII[nd][[1]]},
     Array[Mod[Floor[indexP/(2 nr)], nr] &, Table[2 nr^2, {nd}]]]
\end{verbatim}
}

The following function implements \hyperlink{constmilewgconst}{Construction \Roman{constmilewg}}. It takes as input the 
number of roots of unity, {\tt nr}, the $k$ parameter, {\tt k}, and the number of dimensions, {\tt nd} and returns the 
multi--dimensional perfect array, \textbf{S}: 
{\small
\begin{verbatim}
In[14]:= indexIX[nr_?EvenQ, k_Integer, nd_?EvenQ] := 
      Function @@ {Sum[Slot[n] (Slot[n] + 1) nr^k/2, {n, 1, nd/2}] + index[nd][[1]]}

In[15]:= indexIX[nr_?OddQ, k_Integer, nd_?EvenQ] := 
      Function @@ {Sum[Slot[n] (Slot[n] + 1) nr^k, {n, 1, nd/2}] + index[nd][[1]]}

In[16]:= ConstructionIX[nr_Integer, k_Integer, nd_?EvenQ] := With[{indexF = indexIX[nr, k, nd]}, 
  ArrayFlatten[
      Mod[Array[indexF, Table[nr^(k + 1), {nd/2}]~Join~Table[nr^k, {nd/2}], 0], 
          nr^(k + 1)], nd/2]]
\end{verbatim}
}

The array plot below shows the beautiful structure of a perfect $1024\times 1024$ array over $64$ roots of unity:
{\small
\begin{verbatim}
In[17]:= ArrayPlot[ConstructionIX[4, 2, 4], Frame -> False]
\end{verbatim}
\verb~Out[17]= ~ \vcenteredinclude{width=8cm}{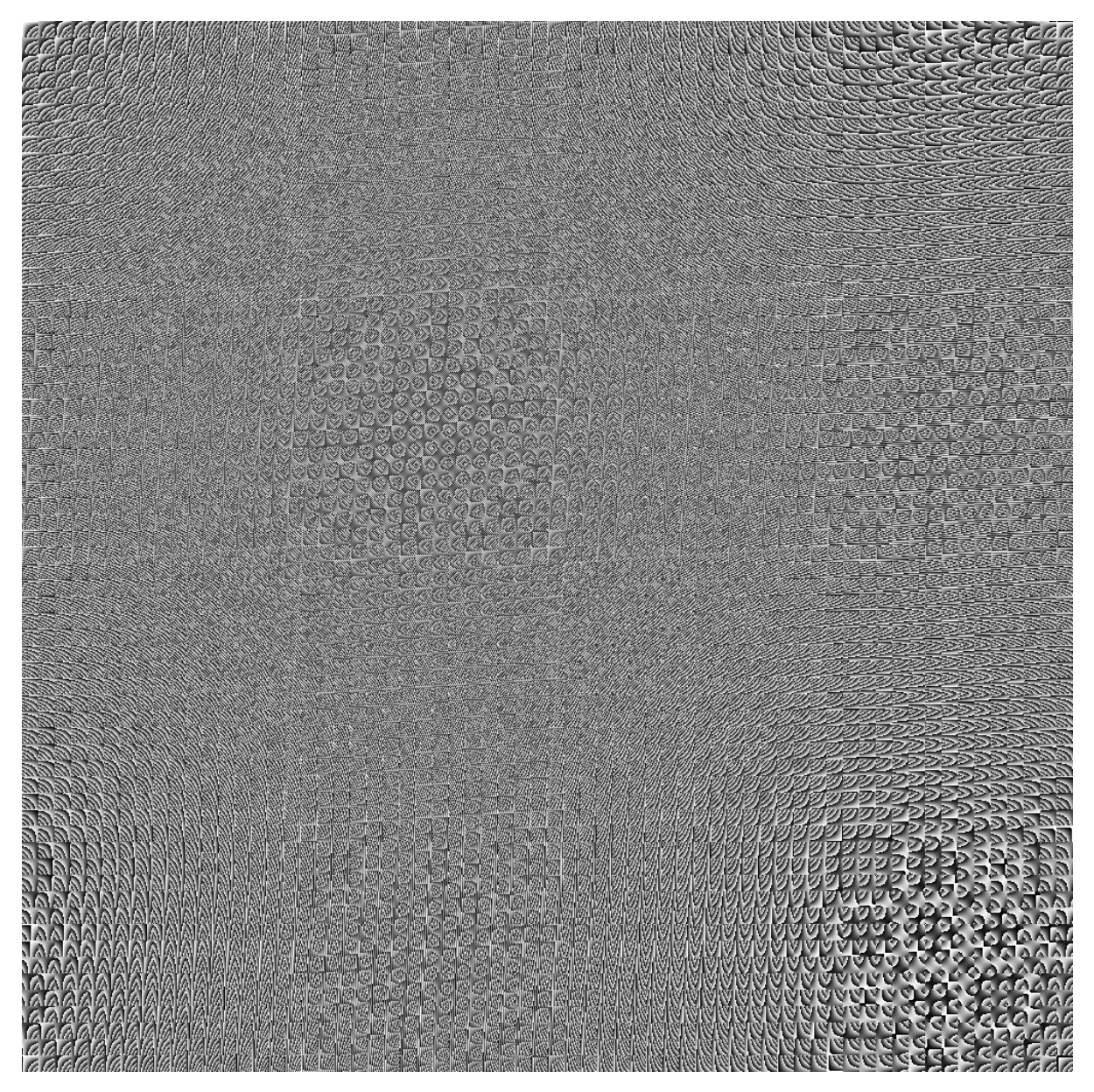}}\\

\end{document}